\documentclass[runningheads]{llncs}

\usepackage{fullpage}

\usepackage{graphicx}

\usepackage[T1]{fontenc}
\usepackage{jybyPython}
\usepackage{hyperref}

\usepackage{epsfig}
\usepackage{algorithm} 

\usepackage{versions}
\excludeversion{INUTILE}
\excludeversion{TODO}
\includeversion{LONG}\excludeversion{SHORT}
\excludeversion{CONDITIONALLOWERBOUND}

\begin{document}

\title{Adaptive Computation of the \\ Discrete Fr\'echet Distance}
\author{J\'er\'emy Barbay}
\authorrunning{J. Barbay}   
%
\tocauthor{J\'er\'emy Barbay (Universidad de Chile)}
\institute{
  Departamento de Ciencias de la Computaci{\'o}n, \\
  University of Chile,  \\
\email{jeremy@barbay.cl}
}

\maketitle              


\begin{abstract}
The discrete Fr{\'e}chet distance is a measure of similarity between point sequences which permits to abstract differences of resolution between the two curves, approximating the original Fr{\'e}chet distance between curves. Such distance between sequences of respective length $n$ and $m$ can be computed in time within $O(nm)$ and space within $O(n+m)$ using classical dynamic programing techniques, a complexity likely to be optimal in the worst case over sequences of similar lenght unless the Strong Exponential Hypothesis is proved incorrect.  We propose a parameterized analysis of the computational complexity of the discrete Fr{\'e}chet distance in fonction of the area of the dynamic program matrix relevant to the computation, measured by its \emph{certificate width} $\omega$.  We prove that the discrete Fr{\'e}chet distance can be computed in time within $((n+m)\omega)$ and space within $O(n+m+\omega)$.  
\end{abstract}

\begin{INUTILE}
The discrete Fr{\'e}chet distance is a measure of similarity between point sequences which permits to abstract differences of resolution between the two curves, approximating the original Fr{\'e}chet distance between curves. Such distance between sequences of respective length $n$ and $m$ can be computed in time within $O(nm)$ and space within $O(n+m)$ using classical dynamic programing techniques, a complexity likely to be optimal in the worst case over sequences of similar lenght unless the Strong Exponential Hypothesis is proved incorrect.
We propose a parameterized analysis of the computational complexity of the discrete Fr{\'e}chet distance, both in term of its upper bound and of its conditional lower bound, in fonction of the area of the dynamic program matrix relevant to the computation, measured by its \emph{certificate width} $\omega$.  We prove that the discrete Fr{\'e}chet distance can be computed in time within $((n+m)\omega)$ and space within $O(n+m+\omega)$, and that this complexity is likely to be optimal in the worst case over sequences of similar lengths and certificate width, unless the Strong Exponential Hypothesis is proved incorrect.
\end{INUTILE}

\section{Introduction}


Measuring the similarity between two curves has applications in areas such as handwriting recognition~\cite{2007-ICDAR-FrechetDistanceBasedApproachForSearchingOnlineHandwrittenDocuments-Sriraghavendra-Karthik-Bhattacharyya}, protein structure alignment~\cite{2008-JBCB-ProteinStructureStructureAlginmentWithDiscreteFrechetDistance-JiangXuZhu},
quantifying macro-molecular pathways~\cite{2015-ARXIV-PathSimilarityAnalysisAMethodForQuantifyingMacromolecularPathways-SeylerKumarThorpeBeckstein}, morphing~\cite{2002-DCG-NewSimilarityMeasuresBetweenPolylinesWithApplicationsToMorphingAndPolygonSweeping-EfratGuibasHarPeledMitchelMurali}, movement analysis~\cite{2017-EGIS-MovementPatternsInSpatioTemporalData-DudmundssonLaubeWolle}, and many others~\cite{wikipedia:FrechetDistance}.
One of the most popular solutions, the \textsc{Fr{\'e}chet Distance} is a measure of similarity between two curves $\cal P$ and $\cal Q$, that takes into account the location and ordering of the points along the curves.  It permits, among other features, to abstract the difference of resolution between $\cal P$ and $\cal Q$, with application to morphing, handwriting recognition and protein structure alignment, among others~\cite{wikipedia:FrechetDistance}.  
In 1995, Art and Godau~\cite{1995-IJCGA-ComputingTheFrechetDistanceBetweenTwoPolygonalCurves-AltGodau} described an algorithm computing the \textsc{Fr{\'e}chet Distance} between two polygonal curves composed of $n$ and $m$ segments respectively in time within $O(mn\log(mn))$.

One year before (1994), Eiter and Mannila~\cite{1994-TR-ComputingDiscreteFrechetDistance-EiterMannila} had extended the notion of the \textsc{Fr{\'e}chet Distance} between curves to the \textsc{Discrete Fr{\'e}chet Distance} between sequence of points of respective sizes $n$ and $m$, demonstrating that the latter can be used to approximate the former in time within $O(nm)$ and space within $O(n+m)$ using classical dynamic programming techniques.
Two decades later (2014), Bringmann~\cite{2014-FOCS-WhyWalkingTheDogTakesTimeFrechetDistanceHasNoStronglySubquadraticAlgorithmsUnlessSETHFails-Bringmann} showed that this worst case complexity is likely to be optimal, unless a bunch of other problems (among which CNF SAT) can be computed faster than usually expected.
Hence, the bounds about the computational complexity of the \textsc{Discrete Fr\'echet Distance} in the worst case over instances of input sizes $n$ and $m$ are reasonably tight.

Yet, for various restricted classes of curves (e.g.  $\kappa$-bounded, backbone, $ c$-packed and long-edged~\cite{2018-IWISC-FastFrechetDistanceBetweenCurvesWithLongEdges-GudmundssonMirzanezhadMohadesWenk} curves), both the \textsc{Fr\'echet Distance} and the \textsc{Discrete Fr\'echet Distance} are known to be easier to compute.
\begin{INUTILE}
\begin{LONG}
A curve $\cal P$ is \emph{$\kappa$-bounded} if for any two points $x,y\in {\cal P}$, the union of the balls with radii $r$ centered at y and $y$ contains the whole ${\cal P}[x,y]$ where $r$ is equal to $(\kappa/2)$ times the Euclidean distance between $x$ and $y$.  Alt et al. described an algorithm to approximate within a factor of $\kappa$ the \textsc{Discrete Fr\'echet Distance} between two such curves of length at most $n$, in time within $O(n\lg n)$.
A \emph{backbone curve} has constant edge length and requires a minimum distance between non-consecutive vertices. Aronov et al. described an algorithm approximating within a factor of $(1{+}\varepsilon)$ the \textsc{Discrete Fr\'echet Distance} of such curves in time near-linear in the size of the curves.
A curve $\cal P$ is \emph{$c$-packed} if for any ball $B$ the length of the portion of $\cal P$ contained in $B$ is at most $c$ times the diameter of $B$.
\end{LONG}
\end{INUTILE}
%
Among other examples, we consider the \textsc{Fr\'echet Distance Decision problem}, which consists in deciding whether the \textsc{Fr\'echet Distance} between two curves is equal to a given value $f$.
In 2018, Gudmundsson et al.~\cite{2018-IWISC-FastFrechetDistanceBetweenCurvesWithLongEdges-GudmundssonMirzanezhadMohadesWenk} described an algorithm deciding if the \textsc{Fr\'echet Distance} is equal to a given value $f$ in time linear in the size of the input curves when each edge is longer than the \textsc{Fr\'echet Distance} between those two curves.
\begin{LONG}
Based on this algorithm, they also showed how to approximate the \textsc{Fr\'echet Distance} within a factor of $\sqrt{d}$ in linear time, and preprocessed in linear time in order to support decision queries in time within $O(m\log^2 n)$.
\end{LONG}
Those results easily project to the \textsc{Discrete Fr\'echet Distance}.

%
%
Those results for the mere computation of the \textsc{Discrete Fr\'echet Distance} suggest that one does not always need to compute the $n\times m$ values of the dynamic program. \textbf{Can such approaches be applied to more general instances}, such that the \textbf{area of the dynamic program which \emph{needs} to be computed} measures \textbf{the \emph{difficulty} of the instance}?


In this work we perform a parameterized analysis of the computational complexity of the \textsc{Discrete Fr{\'e}chet Distance}, in function of the area of the dynamic program matrix relevant to the computation, measured by its \emph{certificate width} $\omega$.
After describing summarily the traditional way to compute the \textsc{Discrete Fr\'echet Distance}\begin{CONDITIONALLOWERBOUND}, the outline of the recent conditional lower bound on its computational complexity\end{CONDITIONALLOWERBOUND} and the particular case of long edged curves (Section~\ref{sec:preliminaries}), we describe an optimization of the classical dynamic program based on two simple techniques, banded dynamic programming and thresholding (Section~\ref{sec:adaptiveDynamicProgram}), and we prove that this program runs in time within $O((n+m)\omega)$ and space within $O(n+m)$ (Section~\ref{sec:complexityAnalysis}). \begin{CONDITIONALLOWERBOUND} This complexity is likely to be optimal in the worst case over sequences of similar lengths and certificate width, if the Strong Exponential Time Hypothesis is correct (Section~\ref{sec:parameterizedConditionalLowerBound}).\end{CONDITIONALLOWERBOUND} We conclude with a discussion in Section~\ref{sec:discussion} of how our results generalize those of Gudmundsson et al.~\cite{2018-IWISC-FastFrechetDistanceBetweenCurvesWithLongEdges-GudmundssonMirzanezhadMohadesWenk}, and the potential applications and generalizations of our techniques to other problems where dynamic programs have given good results

\section{Preliminaries}
\label{sec:preliminaries}

Before describing our results, we describe some classical results upon which we build: the classical dynamic program computing the \textsc{Discrete Fr\'echet Distance}\begin{CONDITIONALLOWERBOUND}, the conditional lower bound\end{CONDITIONALLOWERBOUND}, and the ``easy'' case of long-edged curves described by Gudmundsson et al.~\cite{2018-IWISC-FastFrechetDistanceBetweenCurvesWithLongEdges-GudmundssonMirzanezhadMohadesWenk}. 

\subsubsection{Classical Dynamic Program}
\label{sec:classical}

Let $P[1..n]$ and $Q[1..m]$ be sequences of $n$ and $m$ points with $n\geq m$. The \textsc{Discrete Fr\'echet Distance} between $P$ and $Q$ is the minimal width of a traversal of $P$ and $Q$, where the width of a traversal is the maximal distance separating two points $u\in P$ and $v\in Q$, where $u$ and $v$ progress independently, but always forward.

Such a distance is easily computed using classical techniques from dynamic programming. 
\begin{LONG}
The distance between $P[1..n]$ and $Q[1..m]$ can be reduced in constant time to the minimum between the distance between $P[1..n-1]$ and $Q[1..m]$, the distance between $P[1..n]$ and $Q[1..m-1]$, and the distance between $P[1..n-1]$ and $Q[1..m-1]$. 
\end{LONG}
Algorithm~\ref{algo:classical} (page~\pageref{algo:classical}) describes a simple implementation in \texttt{Python}, executing in time within $O(nm)$.
\begin{algorithm}
\caption{Classical algorithm to compute the \textsc{Discrete Fr\'echet Distance} between two sequences of points $P$ and $Q$. The implementation is decomposed in two parts:  the \texttt{computation} function initializes the array of values, which is filled recursively by the \texttt{recursive} function. For the sake of space, the documentation strings and unit tests were not included, but the source file including those is available at \url{https://github.com/FineGrainedAnalysis/Frechet}.
  \label{algo:classical}}
\begin{python}
def recursive(dpA,i,j,P,Q):
    if dpA[i,j] > -1:
        return dpA[i,j]
    elif i == 0 and j == 0:
        dpA[i,j] = distance(P[0],Q[0])
    elif i > 0 and j == 0:
        dpA[i,j] = max(
            recursive(dpA,i-1,0,P,Q),
            distance(P[i],Q[0]))
    elif i == 0 and j > 0:
        dpA[i,j] = max(
            recursive(dpA,0,j-1,P,Q),
            distance(P[0],Q[j]))
    elif i > 0 and j > 0:
        dpA[i,j] = max(
            min(
                recursive(dpA,i-1,j,P,Q),
                recursive(dpA,i-1,j-1,P,Q),
                recursive(dpA,i,j-1,P,Q)),
            distance(P[i],Q[j]))
    else:
        dpA[i,j] = float("inf")
    return dpA[i,j]

def computation(P,Q):
    dpA = numpy.ones((len(P),len(Q)))
    dpA = numpy.multiply(dpA,-1)
    d = recursive(dpA,len(P)-1,len(Q)-1,P,Q)
    return d
\end{python}
\end{algorithm}
While such a simple algorithm also requires space within $O(nm)$, a simple optimization yields a space within $O(n+m)$, by computing the \textsc{Discrete Fr\'echet Distance} between $P[1..i]$ and $Q[1..j]$ for increasing values of $i$ and $j$, one column and row at the time, keeping in memory only the previous column and row.
We describe in Section~\ref{sec:adaptiveDynamicProgram} a more sophisticated algorithm which avoids computing some of the $n\times m$ values computed by Algorithm~\ref{algo:classical}.

\begin{CONDITIONALLOWERBOUND}
\subsubsection{Conditional Lower Bound}
\label{sec:conditionalLowerBound}

In 2005, Williams~\cite{2005-TCS-ANewAlgorithmForOptimal2ConstraintSatisfactionAndItsImplications-Ryan} related the computational complexity of some general NP-hard problems with the parameterized computational complexity of some problems in computational geometry (in particular, the \textsc{Orthogonal Vector} problem in arbitrary dimension $d$), usually solved in polynomial time without a tight computational lower bound.
Over the years, similar relations have been described between other problems~\cite{2014-SODA-FindingOrthogonalVectorsInDiscreteStructures-WilliamsYu} beyond computational geometry, such as the \textsc{String Edit Distance}~\cite{2015-STOC-EditDistanceCannotBeComputedInStronglySubquadraticTimeUnlessSETHIsFalse-BackursIndyk}, and the \textsc{Fr\'echet Distance}~\cite{2014-FOCS-WhyWalkingTheDogTakesTimeFrechetDistanceHasNoStronglySubquadraticAlgorithmsUnlessSETHFails-Bringmann}.

Bringmann~\cite{2014-FOCS-WhyWalkingTheDogTakesTimeFrechetDistanceHasNoStronglySubquadraticAlgorithmsUnlessSETHFails-Bringmann} showed that the worst case complexity of $O(nm)$ is likely to be optimal, as a better computational complexity would imply an algorithm for \textsc{CNF SAT} running in much faster time than usually expected.

\begin{TODO}
DESCRIBE the technique used by Bringmann.
\end{TODO}
\end{CONDITIONALLOWERBOUND}

\subsubsection{Easy instances of the Fr\'echet Distance}
\label{sec:FrechetDistanceBetweenCurvesWithLongEdges}

For various restricted classes of curves, such as 
\begin{LONG}
$\kappa$-bounded, backbone, $c$-packed and
\end{LONG}
long-edged~\cite{2018-IWISC-FastFrechetDistanceBetweenCurvesWithLongEdges-GudmundssonMirzanezhadMohadesWenk} curves, both the \textsc{Fr\'echet Distance} and the \textsc{Discrete Fr\'echet Distance} are known to be easier to compute (or approximate).
\begin{LONG}

A curve $\cal P$ is \emph{$\kappa$-bounded} if for any two points $x,y\in {\cal P}$, the union of the balls with radii $r$ centered at y and $y$ contains the whole ${\cal P}[x,y]$ where $r$ is equal to $(\kappa/2)$ times the Euclidean distance between $x$ and $y$.  Alt et al. described an algorithm to approximate within a factor of $\kappa$ the \textsc{Discrete Fr\'echet Distance} between two such curves of length at most $n$, in time within $O(n\lg n)$.

A \emph{backbone curve} has constant edge length and requires a minimum distance between non-consecutive vertices. Aronov et al. described an algorithm approximating within a factor of $(1+\varepsilon)$ the \textsc{Discrete Fr\'echet Distance} of such curves in time near-linear in the size of the curves.

A curve $\cal P$ is \emph{$c$-packed} if for any ball $B$ the length of the portion of $\cal P$ contained in $B$ is at most $c$ times the diameter of $B$. 

\end{LONG}
%
In 2018, Gudmundsson et al.~\cite{2018-IWISC-FastFrechetDistanceBetweenCurvesWithLongEdges-GudmundssonMirzanezhadMohadesWenk} showed that in the special case where all the edges of the polygonal curve are longer than the \textsc{Fr\'echet Distance}, the latter can be decided (i.e., checking a value of the \textsc{Fr\'echet Distance}) in linear time in the size of the input,\begin{SHORT} and \end{SHORT} computed in time within $O((n+m)\lg(n+m))$\begin{LONG}, approximated within a factor of $\sqrt{d}$ in linear time in the size of the input, and preprocessed in linear time in order to support decision queries in time within $O(m\log^2 n)$\end{LONG}.

In the next section, we describe a quite simple algorithm which not only takes advantage of long edged curves, but of any pair of curves for which a consequent part of the array of the dynamic program can be ignored.

\section{An opportunistic Dynamic Program}
\label{sec:adaptiveDynamicProgram}

We describe an algorithm based on two complementary techniques: 
first, a \emph{banded dynamic program}, which \emph{approximates} the value computed by a classical dynamic program by considering only the values of the dynamic program within a range of width $w$ (for some parameter $w$) around the diagonal (a technique previously introduced for the computation of the \textsc{Edit Distance} between two strings); and 
second, a \emph{thresholding} process, which \emph{accelerates} the computation by cutting the recurrence any time the distance computed becomes larger or equal to a threshold $t$ (for some parameter $t$ corresponding to a distance already achieved by some traversal of the two curves).
The combination of those two techniques, combined with a parametrization of the problem, yields the parameterized upper bound on the computational complexity of the \textsc{Discrete Fr\'echet Distance}.

\subsubsection{Banded Dynamic Program:}
\label{sec:bandedDynamicProgram}

When computing the \textsc{Edit Distance} (e.g., the \textsc{Delete Insert Edit Distance}, or the \textsc{Levenshtein Edit Distance}~\cite{2000-SPIRE-ASurveyOfLongestCommonSubsequenceAlgorithms-BergrothHakonenRaita}) between similar strings $S\in[1..\sigma]^n$ and $T\in[1..\sigma]^m$ (i.e., their \textsc{Edit Distance} $d$ is small), it is possible to compute less than $n\times m$ cells of the dynamic program array, and hence compute the \textsc{Edit Distance} in time within $O(d(n{+}m))\subseteq O(nm)$.
The ``trick'' is based on the following observation: when the distance between the two strings is $d$, the ``paths'' corresponding to $d$ operations transforming $S$ into $T$ in the matrix of the dynamic program errs at most at distance $d$ from the diagonal between the cell $(1,1)$ and the cell $(n,m)$.
Based on this observation, it is sufficient to compute  the number of operations corresponding to paths inside a ``band'' of width $d$ around such a diagonal~\cite{2000-SPIRE-ASurveyOfLongestCommonSubsequenceAlgorithms-BergrothHakonenRaita}.
\begin{TODO}
CHECK that ~\cite{2000-SPIRE-ASurveyOfLongestCommonSubsequenceAlgorithms-BergrothHakonenRaita} is indeed describing the technique of Banded Dynamic Program!
\end{TODO}
This technique needs some adaptation to be applied to the computation of the \textsc{Discrete Fr\'echet Distance} $f$ between two curves, for two reasons: 
first, $f$ is a real number (it corresponds to the Euclidean distance between two points) and not an integer as the number of edition operations, and this number is independent of the number of cells of the dynamic program being computed; and 
second, the definition of the \textsc{Discrete Fr\'echet Distance} is based on a maximum rather than a sum, which actually makes another optimization possible, described in the next paragraph.

\subsubsection{Thresholding:}
\label{sec:thresholding}

Given two sequences of points $P[1..n]$ and $Q[1..m]$, consider the \emph{Euclidean matrix} $E(P,Q)$ of all $n\times m$ distances between a point $u\in P$ and a point $v\in Q$.  Any parallel traversal of $P$ and $Q$ corresponds to a path in $E(P,Q)$ from the top left cell $(1,1)$ to the bottom right cell $(n,m)$; the width of such a traversal is the maximum value in $E(P,Q)$ on this path; and the \textsc{Discrete Fr\'echet Distance} is the minimum width achieved over all such traversals.

Now suppose that, as for the \textsc{Edit Distance} between two similar strings, the traversal of the Euclidean matrix $E(P,Q)$ corresponding to the \textsc{Discrete Fr\'echet Distance} $f$ between $P$ and $Q$ is close to the diagonal from $(1,1)$ to $(n,m)$, and that any traversal diverging from such a path ``encounters'' a pair of points $(u,v)$ at euclidean distance larger than $f$ (in particular, this happens when the two curves are ``long edged'' compared to $f$). Then, some of the values of the cells of the dynamic program matrix outside of this diagonal can be ignored for the computation of the \textsc{Discrete Fr\'echet Distance} between $P$ and $Q$.

In the following paragraph we describe  how to combine those two techniques into an adaptive algorithm taking advantage of ``easy'' instances where a large quantity of cells of the dynamic program can be ignored.

\subsubsection{Combining the two techniques:}
\label{sec:Combination}

The solution described consists of two algorithms: an approximation algorithm~\ref{algo:approximation} which computes a parameterized upper bound on the \textsc{Discrete Fr\'echet Distance}, and a computation algorithm~\ref{algo:adaptive} which calls the previous one iteratively with various parameter values, in order to compute the real \textsc{Discrete Fr\'echet Distance} of the instance.

Algorithm~\ref{algo:approximation} lists an implementation in \texttt{Python} of an algorithm which, given as parameters
two arrays of points $P$ and $Q$,
an \texttt{integer} \emph{width} $w$, and
a \texttt{float} \emph{threshold} $t$; 
computes an upper bound of the \textsc{Discrete Fr\'echet Distance} between $P$ and $Q$, obtained by computing only the cells within a band of width $2w$ around the diagonal from the top left cell $(1,1)$ to the bottom right cell $(n,m)$, and cutting all sequences of recursive calls when reaching a distance of value $t$ or above. 
This algorithm uses space within $(n+m)$ as it computes the values from $(1,1)$ to $(n,m)$ by updating and switching alternatively two arrays of size $n$ and two arrays of size $m$ (respectively corresponding to rows and columns of the dynamic program matrix).
Its running time is within $O(w(n+m))$, as it computes at most $w(n+m)$ cells of the dynamic program array.
Furthermore, it not only returns the value of the upper bound computed, but also a Boolean \texttt{breached} indicating if the border of the banded diagonal has been reached during this computation. When such a border has \emph{not} been reached (and the threshold value $t$ is smaller than or equal to the \textsc{Discrete Fr\'echet Distance} between $P$ and $Q$), the value returned is equal to the \textsc{Discrete Fr\'echet Distance} between $P$ and $Q$.

\begin{algorithm}
\caption{Parameterized Algorithm to approximate the \textsc{Discrete Fr\'echet Distance} between two sequences of points by computing only values of the dynamic program within a band of width $w$ around the diagonal, and limiting the recursion to distances smaller than a threshold $t$.
  \label{algo:approximation}}
\begin{python}
def approximation(P,Q,w,t):
    bReached = False
    n = len(P)
    m = len(Q)
    assert( m <= n )
    assert( m > 0 )
    def e(i,j):
        d = utils.distance(P[i],Q[j])
        if d < t:
            if (i-j) >= w or (j-i) >= w:
                bReached = True
            return d
        else:
            return float("inf")
    oldRow = np.ones(n)
    oldColumn = np.ones(m)
    oldRow[0] = oldColumn[0] = e(0,0)
    for s in range(1,m):
        newRow = np.ones(n)
        for i in range(max(1,s-w+1),s):
            newRow[i]  = max(e(i,s),min(
                oldRow[i],
                oldRow[i-1],
                newRow[i-1]))
        newColumn = np.ones(m)
        for j in range(max(1,s-w+1),s):
            newColumn[j]  = max(e(s,j),min(
                newColumn[j-1],
                oldColumn[j-1],
                oldColumn[j]))
        newColumn[s] = newRow[s] = max(e(s,s),min(
            newRow[s-1],
            newColumn[s-1],
            oldRow[s-1]))
        oldRow = newRow
        oldColumn = newColumn
    for s in range(m,n):
      newColumn = np.ones(m)
      for j in range(max(1,s-w+1),m):
          newColumn[j]  = max(e(s,j),min(
              oldColumn[j],
              oldColumn[j-1],
              newColumn[j-1]))
      oldColumn = newColumn
    return bReached,newRow[n-1]
\end{python}
\end{algorithm}

Algorithm~\ref{algo:adaptive} lists an implementation in \texttt{Python} of an algorithm which, given as parameters two sequences of points $P$ and $Q$, calls the approximation Algorithm~\ref{algo:approximation} on $P$ and $Q$ for widths of exponentially increasing value (by a factor of two). The first call is performed with an infinite threshold (no information is available on the similarity of the curve at this point), but each successive calls use the best upper bound on the \textsc{Discrete Fr\'echet Distance} between $P$ and $Q$ previously computed as a threshold.
\begin{algorithm}
\caption{Adaptive algorithm to compute the \textsc{Discrete Fr\'echet Distance} between two sequences of points, by iteratively approximating it with increasing width, using the value of the previous approximation to potentially reduce the number of distances being computed.
  \label{algo:adaptive}}
\begin{python}
def computation(P,Q):
    if len(P)<len(Q): 
      P,Q = Q,P
    if len(Q) == 0:
       return float("inf")
    w = 1
    bReached,t=approximation(P,Q,w,float("inf"))
    while bReached and w < len(Q):
      w = 2*x
      bReached,t=approximation(P,Q,w,t)
    return t
\end{python}
\end{algorithm}

\begin{LONG}
The intuition of the correctness of Algorithm~\ref{algo:adaptive} is trivial: on the last execution of Algorithm~\ref{algo:approximation}, either all the values of the dynamic program array were computed, or all recursions where stopped before they reached a value which is not computed (because out of the bandwidth). We formalize the argument in the following theorem.
\begin{theorem}
Algorithm~\ref{algo:adaptive} correctly returns the \textsc{Fr\'echet Distance} of the two input sequences of points.
\end{theorem}
\begin{proof}
Consider two sequences of points $P$ and $Q$ given as input to Algorithm~\ref{algo:adaptive}.
On the last execution of Algorithm~\ref{algo:approximation},
\begin{itemize}
\item the threshold value corresponds to a viable traversal of the Fr\'echet matrix, and
\item either all the values of the dynamic program array were computed ($w=n$), or all recursions where stopped before they reached a value which is not computed ($\mathtt{bReached}=\mathtt{False}$)).
\end{itemize}
As a consequence, the computation of the band of width $w$ correctly simulates the computation of the whole dynamic program, from which the correctness of the algorithm ensues.
\end{proof}
\end{LONG}
In the next section, we analyze the running time of Algorithm~\ref{algo:adaptive}, and describe a parameterized upper bound on it.

\section{Parameterized  Upper Bound}
\label{sec:complexityAnalysis}

The running time of the approximation Algorithm~\ref{algo:approximation} when given parameter $w$ is clearly within $O(w(n+m))$: it computes within $O(w)$ cells in at most $n+m$ rounds, each in constant time. A finer upper bound taking into account the value of the parameter $t$ requires more hypothesis on the relation between $P$ and $Q$, for which we need to consider the running time of the computation Algorithm~\ref{algo:adaptive}.
We  model such hypothesis on the instance in the form of a \emph{certificate}, and more specifically in the form of a \emph{certificate area} of the Euclidean matrix corresponding to a set of values which suffice to \emph{certify} the value of the \textsc{Discrete Fr\'echet Distance}.
\begin{definition}
Given two sequences of points $P[1..n]$ and $Q[1..m]$ of respective lengths $n$ and $m$ and of \textsc{Discrete Fr\'echet Distance} $f$, a \emph{Certificate Area} of the instance formed by $P$ and $Q$ is an area of the Euclidean matrix of $P$ and $Q$ containing both $(1,1)$ and $(n,m)$, and delimited by two paths (one above and one below), both such that the minimum value on this path is larger than or equal to $f$. The width of such a certificate area is the minimal width of a banded diagonal containing both paths.
\end{definition}

The surface of such an area is a good measure of the difficulty to certify the \textsc{Discrete Fr\'echet Distance}, but the minimal width of such an area lends itself better to an analysis of the running time of the computation Algorithm~\ref{algo:adaptive}:
\begin{definition}
Given two sequences of points $P[1..n]$ and $Q[1..m]$, the \emph{Certificate Width} $\omega$ of $(P,Q)$ is the minimum width of a certificate area, taken over all possible certificate areas of $(P,Q)$. 
\end{definition}

Such a width can be as large as $n+m$ in the worst case over instances formed by sequences of points of respective lengths $n$ and $m$, but the smaller this \emph{certificate width} is , the faster Algorithm~\ref{algo:adaptive} runs:
\begin{theorem}
Given two sequences of points $P[1..n]$ and $Q[1..m]$ forming an instance of certificate width $\omega$, Algorithm~\ref{algo:adaptive} computes the \textsc{Discrete Fr\'echet Distance} between $P$ and $Q$ in time within $O((n{+}m)\omega)$ and space within $O(n{+}m)$.
\end{theorem}

\begin{CONDITIONALLOWERBOUND}
\section{Parameterized Conditional Lower Bound}
\label{sec:parameterizedConditionalLowerBound}

\end{CONDITIONALLOWERBOUND}

Beyond the necessity to measure experimentally the certificate width of practical instances of the \textsc{Discrete Fr\'echet Distance}, and the exact running time of Algorithm~\ref{algo:adaptive} on such instances, we discuss some more subtle options for future work in the next section.

\section{Discussion}
\label{sec:discussion}

The results described in this work are by far only preliminary. Among the various questions that those preliminary results raise, we discuss here the relation to the long edged sequences  recently described by Gudmundsson et al.~\cite{2018-IWISC-FastFrechetDistanceBetweenCurvesWithLongEdges-GudmundssonMirzanezhadMohadesWenk}; 
a potential parameterized conditional lower bound matching our parameterized upper bound on the computational complexity of the \textsc{Discrete Fr\'echet Distance};
(the not so) similar results on the \textsc{Orthogonal Vector} decision problem; and 
the possibility of a theory of reductions between parameterized versions of polynomial problems without clear (parameterized or not) computational complexity lower bounds. 

\subsubsection{Relation to Long Edged Sequences:}
\label{sec:relationToLongEdgedSequences}

In 2018, Gudmundsson et al.~\cite{2018-IWISC-FastFrechetDistanceBetweenCurvesWithLongEdges-GudmundssonMirzanezhadMohadesWenk} described an algorithm deciding if the Fr\'echet distance is equal to a given value $f$ in time linear in the size of the input curves when each edge is longer than the \textsc{Fr\'echet Distance} between those two curves.
Algorithm~\ref{algo:adaptive} is more general than Gudmundsson et al.'s algorithm~\cite{2018-IWISC-FastFrechetDistanceBetweenCurvesWithLongEdges-GudmundssonMirzanezhadMohadesWenk}, but it also performs in linear time on long-edged instances: the traversal corresponding to the \textsc{Fr\'echet Distance} of such an instance is along the diagonal, implying a \emph{certificate width} of $1$. See Figures~\ref{fig:longEdgedCurvesEuclideanMatrix},~\ref{fig:longEdgedCurvesFrechetMatrix} and~\ref{fig:longEdgedCurvesDynamicProgramMatrix} for the Euclidean matrix, Fr\'echet Matrix and Dynamic Program Matrix of a random instance formed of $5$ points, each edge of length $100$ with a \textsc{Fr\'echet Distance} of $13.45$  (see Appendix~\ref{sec:randomGenerationOfLongEdgedCurves} for the \texttt{Python} code used to generate long edged instances).

The ratio between the \textsc{Fr\'echet Distance} and the minimal edge length of the curves might prove to be a more ``natural'' parameter than the \emph{certificate width} to measure the ``difficulty'' of computing the \textsc{Fr\'rechet Distance} of a pair of curves: we focused on the \emph{certificate width} in the hope that such a technique can find applications in the analysis of other problems on which dynamic programming has yield good solutions.

\begin{figure}
\begin{minipage}[t]{.45\linewidth}
$$\begin{array}{  *{ 6 }{c  }} 
\mathbf{1.41} &	101.43         & 193.97        & 294.66       &	199.17         & 227.36       \\ 
100.14        &	\mathbf{13.45} & 94.67         & 195.37       &	99.5           & 142.89       \\ 
199.48        &	97.99          & \mathbf{5.39} & 105.43       &	6.13           & 113.25       \\ 
290.44        &	192.56         & 104.04        & \mathbf{6.0} &	97.26          & 109.6        \\ 
193.23        &	93.14          & 13.17         & 104.98       &	\mathbf{10.44} & 99.05        \\ 
232.69        &	156.27         & 112.89        & 104.58       &	107.64         & \mathbf{6.4} \\ 
\end{array}$$
\caption{Euclidean Matrix for a Long Edged Instance: the $6$ points from the first curve were randomly generated at distance $100$ of each other, while the points from the second curve were generated by perturbing within a distance of $10$ from the points of the first curve.
\label{fig:longEdgedCurvesEuclideanMatrix}}
\end{minipage}
\hfill
\begin{TODO}
FIND how to plot the curves of 5 points!
\end{TODO}
\begin{minipage}[t]{.45\linewidth}
$$\begin{array}{  *{ 6 }{c  }} 
\mathbf{1.41} &	101.43         & 193.97 	& 294.66         & 294.66         & 294.66         \\ 
100.14        &	\mathbf{13.45} & 94.67          & 195.37         & 195.37         & 195.37         \\ 
199.48        &	97.99          & \mathbf{13.45} & 105.43         & 105.43         & 113.25         \\ 
290.44        &	192.56         & 104.04 	& \mathbf{13.45} & 97.26          & 109.6          \\ 
290.44        &	192.56         & 104.04 	& 104.98         & \mathbf{13.45} & 99.05          \\ 
290.44        &	192.56         & 112.89 	& 104.58         & 107.64         & \mathbf{13.45} \\ 
\end{array}$$
\caption{Fr\'echet Matrix for the same Long Edged Instance as Figure~\ref{fig:longEdgedCurvesEuclideanMatrix}: the traversal corresponding to the \textsc{Fr\'echet Distance} of the instance is along the diagonal (highlighted in bold here), resulting in a \textsc{Fr\'echet Distance} of $13.45$. \label{fig:longEdgedCurvesFrechetMatrix}}
\end{minipage}
\hfill
\begin{minipage}[t]{.45\linewidth}
$$\begin{array}{  *{ 6 }{c  }} 
\mathbf{1.41} &	inf            & inf            & -1.0          & -1.0          & -1.0          \\ 
inf           &	\mathbf{13.45} & inf            & inf            & -1.0          & -1.0          \\ 
inf           &	inf            & \mathbf{13.45} & inf            & 6.13           & -1.0          \\ 
-1.0         &	inf            & inf            & \mathbf{13.45} & inf            & inf            \\ 
-1.0         &	-1.0          & 13.17          & inf            & \mathbf{13.45} & inf            \\ 
-1.0         &	-1.0          & -1.0          & inf            & inf            & \mathbf{13.45} \\ 
\end{array}$$
\caption{Dynamic Program Matrix for the same Long Edged Instance as Figure~\ref{fig:longEdgedCurvesEuclideanMatrix}, with width $3$ and threshold $20$: ``inf'' denotes interrupted recurrences because the distance found is already larger than the threshold, meanwhile values outside of the band of width $3$ are marked with ``-1''.
 \label{fig:longEdgedCurvesDynamicProgramMatrix}}
\end{minipage}
\hfill
\begin{minipage}[t]{.45\linewidth}
$$\begin{array}{  *{ 6 }{c  }} 
\mathbf{ 9.43 } 	&	19.48 	&	19.48 	&	-1.0 	&	-1.0 	&	-1.0 	\\ 
18.81 	&	\mathbf{ 11.31 } 	&	11.31 	&	16.86 	&	-1.0 	&	-1.0 	\\ 
18.81 	&	14.26 	&	\mathbf{ 11.31 } 	&	11.31 	&	inf 	&	-1.0 	\\ 
-1.0 	&	inf 	&	16.07 	&	\mathbf{ 11.31 } 	&	16.5 	&	13.77 	\\ 
-1.0 	&	-1.0 	&	inf 	&	11.31 	&	\mathbf{ 11.31 } 	&	11.31 	\\ 
-1.0 	&	-1.0 	&	-1.0 	&	3.17 	&	14.06 	&	\mathbf{ 11.31 } 	\\ 
\end{array}$$
\caption{Dynamic Program Matrix for a general instance. 
The  $6$ points from the first curve were randomly generated at distance $10$ of each other, the points from the second curve by perturbing within a distance of $10$  the points of the first curve.
The computation of the matrix is performed
with width $3$ and threshold $20$ as before.
\label{fig:generalDynamicProgramMatrix}}
\end{minipage}
\begin{INUTILE}
\hfill
\begin{minipage}[t]{.45\linewidth}
$$\begin{array}{  *{ 6 }{c  }} 
    \mathbf{ 46.0 } 	&	46.0 	&	46.0 	&	-1.0 	&	-1.0 	&	-1.0 	\\ 
    67.0 	&	\mathbf{ 46.0 } 	&	63.0 	&	70.0 	&	-1.0 	&	-1.0 	\\ 
    inf 	&	56.0 	&	\mathbf{ 46.0 } 	&	inf 	&	inf 	&	-1.0 	\\ 
    -1.0 	&	71.0 	&	inf 	&	\mathbf{ 46.0 } 	&	76.0 	&	inf 	\\ 
    -1.0 	&	-1.0 	&	59.0 	&	46.0 	&	\mathbf{ 66.0 } 	&	66.0 	\\ 
    -1.0 	&	-1.0 	&	-1.0 	&	inf 	&	inf 	&	\mathbf{ 66.0 } 	\\ 
  \end{array}$$
\caption{Dynamic Program Matrix for a general instance (not a  Long Edged one). The coordinates of the $6$ points of both curves being generated i.i.d. in a grid of $100\times 100$, while the computation of the matrix is performed with
with width $3$ and threshold $80$, so that to show both the effect of banded dynamic programming and thresholding.
\label{fig:generalDynamicProgramMatrix}}
\end{minipage}
\end{INUTILE}
\end{figure}

\subsubsection{Parameterized Conditional Lower Bound}
\label{sec:parameterizedConditionalLowerBound}

The original motivation for this work was to prove a parameterized conditional lower bound on the computational complexity of the \textsc{Discrete Fr\'echet Distance} as a step-stone for doing the same for the computation of various \textsc{Edit Distances}. The first step in this direction was the identification of a parameter for this problem: the \emph{certificate width}, that seems to be a good candidate. The next step is to refine the reduction from \textsc{CNF SAT} to the \textsc{Discrete Fr\'echet Distance} described by Bringmann~\cite{2014-FOCS-WhyWalkingTheDogTakesTimeFrechetDistanceHasNoStronglySubquadraticAlgorithmsUnlessSETHFails-Bringmann}, in order to define a reduction from (a potential parameterized version of) \textsc{CNF SAT} to a parameterized version of the \textsc{Discrete Fr\'echet Distance}.

\subsubsection{Parameterized Upper and Lower Bound for the  computation of Orthogonal Vectors:}
\label{sec:orthogonalVectors}

Bringmann~\cite{2014-FOCS-WhyWalkingTheDogTakesTimeFrechetDistanceHasNoStronglySubquadraticAlgorithmsUnlessSETHFails-Bringmann} mentions that the reduction from \textsc{SAT CNF} to the computation of the \textsc{Discrete Fr\'echet Distance} is similar to Williams' reduction from \textsc{SAT CNF} to the (polynomial) problem of deciding if two sets of vectors contain an \textsc{Orthogonal Vector} pair, and that there might be a reduction from the \textsc{Orthogonal Vector} decision problem to the computation of the \textsc{Discrete Fr\'echet Distance}. This mention called the \textsc{Orthogonal Vector} decision problem to our attention, and in particular 1) the possibility of a parametrization of the analysis of this problem, and 2) a potential linear (or parameterized) reduction from such a parameterized \textsc{Orthogonal Vector} decision problem to the parameterized computation of the \textsc{Discrete Fr\'echet Distance} described in this work.
It turns out that there exists an algorithm solving the \textsc{Orthogonal Vector} decision problem in time within $O((n{+}m)(\delta+\log(n)+\log(m)))$, where $n$ and $m$ are the respective sizes of the sets of vectors forming the instance, and $\delta$ is the \emph{certificate density} measuring the proportion of cells from the the dynamic program which are sufficient to compute in order to certify the answer to the program. The reduction of this to the \textsc{Discrete Fr\'echet Distance} will be more problematic: the two measures of difficulty seem completely unrelated. 

\subsubsection{A theory of reduction between polynomial parameterized problems}
\label{sec:reductionTheory}

The study of the \emph{parameterized complexity} of NP-hard problems~\cite{2006-BOOK-ParameterizedComplexityTheory-FlumGrohe,2004-CCC-ParameterizedComplexityOfConstraintSatisfactionProblems-Marx} yields a theory of reduction between pairs formed by a decision problem $P$ and a parameter $k$.
The study of \emph{adaptive sorting} algorithms~\cite{1992-ACMCS-ASurveyOfAdaptiveSortingAlgorithms-EstivillCastroWood,1992-ACJ-AnOverviewOfAdaptiveSorting-MoffatPetersson} yields a theory of reductions between parameters measuring the existing disorder in an array to be sorted (which can also be seen as a theory of reductions between pairs of problems and parameters, but where all the problems are equal).
Considering the theory of reductions between polynomial problems such as the \textsc{Discrete Fr\'echet Distance}, the various \textsc{Edit Distances} between strings, the \textsc{Orthogonal Vector} decision problem, and many others, one can imagine that it would be possible to define a theory of reductions between parameterized versions of these problems.

\medskip \textbf{Acknowledgements:} The author would like to thank Travis Gagie for interesting discussions and for pointing out Gudmundsson et al.'s work~\cite{2018-IWISC-FastFrechetDistanceBetweenCurvesWithLongEdges-GudmundssonMirzanezhadMohadesWenk}.
\textbf{Funding:} J\'er\'emy Barbay is partially funded by the project Fondecyt Regular no. 1170366 from Conicyt.
%
%
%
\textbf{Data and Material Availability:}
The source of this article, along with the unabridged code 
will be made publicly available upon publication at the url \url{https://github.com/FineGrainedAnalysis/Frechet}.

\bibliographystyle{splncs04} 
\bibliography{/home/jbarbay/EverGoing/Bibliography/bibliographyDatabaseJeremyBarbay,/home/jbarbay/EverGoing/Publications/publications-ExportedFromOrgmode-Barbay}

\appendix

\section{Additional Algorithms of interest}
\label{sec:additionalAlgorithms}

Even though the source of this article, along with the unabridged code will be made publicly available upon publication at the url \url{https://github.com/FineGrainedAnalysis/Frechet}, we copy here some extracts of the code, which can help the reader to understand the implementation and the experiments performed.

\subsection{Dynamic Program Matrix}
\label{sec:dynamicProgramMatrix}

We describe in Algorithm~\ref{algo:matrixApproximation} the \texttt{Python} implementation of an algorithm to approximate the \textsc{Discrete Fr\'echet \textbf{Matrix}} between two sequences of points, by computing only values of the dynamic program within a band of width $w$ around the diagonal, and limiting the recursion to distances smaller than a threshold $t$. The interest of this algorithm is not its running time, which is within $\Theta(n\times m)$ given that it initializes and returns the whole matrix of the dynamic program, but in the pedagogical value of its output, which yields a visualization of the space explored by Algorithm~\ref{algo:adaptive}.
In particular, this algorithm was used in Section~\ref{sec:relationToLongEdgedSequences} to generate the matrices presented in Figures~\ref{fig:longEdgedCurvesEuclideanMatrix},~\ref{fig:longEdgedCurvesFrechetMatrix} and~\ref{fig:longEdgedCurvesDynamicProgramMatrix}.

\begin{algorithm}
\caption{Adaptive algorithm to approximate the \textsc{Discrete Fr\'echet \textbf{Matrix}} between two sequences of points.
  \label{algo:matrixApproximation}}
\begin{python}
def approximationOfFrechetMatrix(P,Q,w,t):
    bReached = False
    def e(i,j):
        d = distance(P[i],Q[j])
        if d < t:
            if (i-j) >= w or (j-i) >= w:
                bReached = True
            return d
        else:
            return float("inf")
    assert( len(Q) <= len(P) )
    assert( len(Q) > 0 )
    dpA = np.ones((len(P),len(Q)))
    dpA = np.multiply(dpA,-20)
    # Initialize "old" arrays
    dpA[0,0] = e(0,0)
    for i in range(1,min(w,len(P))):
        dpA[i,0] = max(e(i,0),dpA[i-1,0])
    for j in range(1,min(w,len(Q))):
        dpA[0,j] = max(e(0,j),dpA[0,j-1])
    # Compute values in the leftmost square
    for s in range(1,len(Q)):
        for i in range(max(1,s-w+1),s):
            dpA[i,s]  = max(e(i,s),min(
                dpA[i,s-1],
                dpA[i-1,s-1],
                dpA[i-1,s]
            ))
        for j in range(max(1,s-w+1),s):
            dpA[s,j]  = max(e(s,j),min(
                dpA[s-1,j],
                dpA[s-1,j-1],
                dpA[s,j-1]
            ))
        dpA[s,s]  = max(e(s,s),min(
            dpA[s-1,s],
            dpA[s-1,s-1],
            dpA[s,s-1]
        ))
    # Compute values in the rest of the array
    for s in range(len(Q),len(P)):
      for j in range(max(1,s-w+1),len(Q)):
          dpA[s,j]  = max(e(s,j),min(
              dpA[s-1,j],
              dpA[s-1,j-1],
              dpA[s,j-1]
          ))
    return bReached,dpA
\end{python}
\end{algorithm}

\subsection{Random Generation of Long Edged Curves}
\label{sec:randomGenerationOfLongEdgedCurves}

We describe in Algorithm~\ref{algo:generationOfLongEdgedInstances} the \texttt{Python} code used to generate long edged instances in Figures~\ref{fig:longEdgedCurvesEuclideanMatrix},~\ref{fig:longEdgedCurvesFrechetMatrix} and~\ref{fig:longEdgedCurvesDynamicProgramMatrix} of Section~\ref{sec:relationToLongEdgedSequences}.

\begin{TODO}
DEVELOP the description of the algorithm generating long edged instances.
\end{TODO}

\begin{algorithm}
\caption{Code used to generate long edged instances in Figures~\ref{fig:longEdgedCurvesEuclideanMatrix},~\ref{fig:longEdgedCurvesFrechetMatrix} and~\ref{fig:longEdgedCurvesDynamicProgramMatrix} of Section~\ref{sec:relationToLongEdgedSequences}.
  \label{algo:generationOfLongEdgedInstances}}
\begin{python}
def randomPointOnUnitCircle():
    """Generates a random point on the unit circle.
    """
    angleInRadians = random.random() * 2.0 * math.pi
    dx = math.cos(angleInRadians)
    dy = math.sin(angleInRadians)
    return (dx,dy)
def randomLongEdgedCurve(n,edgeLength=10):
    """Generates a random curve of $n$ points 
    such that each edge is of length $edgeLength$

    >>> P = randomCurve(20,10)
    >>> Q = randomlyPerturbedCurve(P,1)

    """
    x,y = randomPointOnUnitCircle()
    P = [(x,y)] 
    for i in range(n):
        dx,dy = randomPointOnUnitCircle()
        x += dx * edgeLength
        y += dy * edgeLength
        P.append((x,y))
    return P
def randomlyPerturbedCurve(P,d=1):
    """Generates a perturbation @Q of a curve @P, 
so that the distance from @P to @Q is not as big 
as if generated separately.

    >>> P = [(1,1),(1,2),(1,3),(1,4),(1,5),(1,6)]
    >>> Q = randomlyPerturbedCurve(P,1)
    """
    Q = []
    for x,y in P:
        xp = x+random.randint(-d, d)
        yp = y+random.randint(-d, d)
        Q.append((xp,yp))
    return Q
\end{python}
\end{algorithm}

\begin{TODO}
\section{Adaptive Solution for the Orthogonal Vector Problem}
\label{sec:orthogonalVector}

\end{TODO}


\end{document}